\documentclass[USenglish]{article}

\usepackage[utf8]{inputenc}
\usepackage[T1]{fontenc}
\usepackage{lmodern}
\usepackage{babel}

\usepackage{float}
\usepackage{caption}
\usepackage{authblk}
\usepackage{geometry}
\usepackage{hyperref}
\usepackage[all]{hypcap}

\hypersetup{
	hidelinks
}

\usepackage{soul}
\usepackage{tikz}
\usepackage{float}
\usepackage{xspace}
\usepackage{xcolor}
\usepackage{amsthm}
\usepackage{amsmath}
\usepackage{amssymb}
\usepackage{graphicx}
\usepackage{longtable}
\usepackage[ruled]{algorithm2e}

\theoremstyle{plain}
\newtheorem{theorem}{Theorem}
\newtheorem{lemma}[theorem]{Lemma}
\newtheorem{corollary}[theorem]{Corollary}

\newcommand{\OO}{\ensuremath{\mathcal{O}}}

\newcommand{\GF}{\mathbb{GF}}

\newcommand{\Expect}{\mathrm{E}}

\newcommand{\eg}{e.g.,\xspace}

\newcommand{\ie}{i.e.,\xspace}

\newcommand{\eos}{\text{~~.}}
\newcommand{\eoss}{~~\text{,}}

\newcommand{\coloneqq}{\mathrel{\mathop:}=}
\newcommand{\floor}[1]{\left\lfloor #1 \right\rfloor}

\floatstyle{ruled}
\newfloat{algorithm}{h}{loa}
\floatname{algorithm}{Algorithm}

\begin{document}
\title{Cyclone Codes}

\author[1]{Christian Schindelhauer}
\author[2]{Andreas Jakoby}
\author[3]{Sven Köhler}
\affil[1]{University of Freiburg, Germany\\
	\texttt{schindel@informatik.uni-freiburg.de}}
\affil[2]{Bauhaus-Universität Weimar, Germany\\
	\texttt{andreas.jakoby@uni-weimar.de}}
\affil[3]{University of Freiburg, Germany\\
	\texttt{koehlers@informatik.uni-freiburg.de}}

\maketitle

\begin{abstract}
We introduce Cyclone codes which are rateless erasure resilient codes.
They combine 
Pair codes with Luby Transform (LT) codes by computing  a code symbol from 
a random set of data symbols using bitwise XOR and cyclic shift operations. The number of 
data symbols is chosen according to the Robust Soliton distribution. XOR and 
cyclic shift operations establish a unitary commutative ring if data symbols 
have a length of $p-1$ bits, for some prime number $p$. We consider the graph 
given by code symbols combining two data symbols. If $n/2$ such random pairs 
are given for $n$ 
data symbols, then a giant component appears, which can be resolved in linear time. We can extend Cyclone codes to data symbols of arbitrary even 
length, provided the Goldbach conjecture holds.

Applying results for this giant component, it follows that Cyclone codes have 
the same encoding and decoding time complexity as LT codes, while the overhead 
is upper-bounded by those of LT codes. Simulations indicate that Cyclone codes 
significantly decreases the overhead of extra coding symbols.
\end{abstract}

\section{Introduction}

Currently, video streaming is the dominant source of traffic in the Internet. 
Here, most packets are delivered best effort, which means that packets can be 
erased anytime. The missing packet needs to be resent, which costs a round trip 
time. In result, the video may be halted and rebuffered if not enough packets 
are buffered before-hand. This can be avoided if the network layer (and link 
layer) provides real-time behavior. For this, IPv6 provides special 
quality-of-service flags in order to prioritize media packets. The best 
approach so far is forward error correction with so-called erasure codes.


This is only one of many applications for erasure codes, where additional redundant packets are added such that the original packets can be reconstructed from the remaining packets. RAID disks and long distance satellite communications are other prominent examples. 

\paragraph*{Our Results}
By combining several techniques from previous erasure resilient codes, namely 
Pair codes \cite{ortolf2009paircoding}, Circulant Cauchy codes 
\cite{Schindelhauer2013}, and Luby Transform (LT) codes \cite{Luby02}, we 
present a novel erasure resilient code system called Cyclone codes.

The number of data symbols combined by XOR and cyclic shift operations is 
chosen according 
to the Robust Soliton distribution. XOR and cyclic shift operations establish a unitary commutative 
ring if the data symbol has length $p-1$ for some prime number $p\geq 3$. We 
consider the graph induced 
by code symbols describing two data symbols. If $n/2$ such random pairs are 
given for $n$ data symbols, 
then the giant component appears \cite{erdHos1959random}, which can be resolved 
using these operations. We can extend Cyclone codes to data symbols of 
arbitrary even length provided the Goldbach conjecture holds, which has been 
shown for any reasonable symbol length.


Cyclone codes
are rateless, non-systematic, memoryless erasure codes which share their 
asymptotic coding and decoding time and asymptotic coding overhead with LT 
codes. Yet, simulations show that it considerably improves the coding overhead 
by at least 40\% compared to LT codes.

%
%



\section{Model}

We assume a sender and receiver connected by a faulty communication
channel in the \emph{packet erasure model}: (a) packets are delivered in the 
correct order and (b) individual packets are either correctly delivered or the 
receiver is notified about their loss.

Data symbols (the input to the erasure code) are denoted by
$x_1, x_2, \ldots, x_n$, where $n$ is the number of available data symbols.
The code symbols (the output of the erasure code) are denoted by $y_1, y_2, \ldots, 
y_m$, where $m$ denotes, depending on context, the number of code symbols 
transmitted or received. Data and code symbols are $w$-bit sequences.

%

For all three erasure codes discussed in this paper, the code symbols are
linear combinations of several data symbols. The indexes of the data symbols 
used as well as the coefficients are random numbers. We assume that sender and 
receiver have access to a shared source of randomness, \ie the receiver can 
reproduce the random numbers drawn by the sender. This can be achieved, for 
example, by using a pseudo-random generator with the same initialization
on both sides. We do not discuss the issue how both sides can agree on the 
same initialization. Also, our analysis neglects questions regarding the time 
needed to generate the (pdeudo-)random numbers or the reliability of their 
randomness.
Note that the packet erasure model allows the receiver to reproduce the
random number used in the computation of every code symbol. Alternatively, the 
sender may embed these parameters in each packet.

\section{Related Work}


For a short history of coding theory we refer to the excellent survey of \cite{4282117}. 
The notion of an erasure channel was introduced in \cite{elias1955coding}. The standard method for 
erasure codes are Reed-Solomon codes \cite{Reed:1960:PCC}. In 
\cite{BloemerEtAlAnXorBasedErasureResilient1995},
Cauchy matrices have been proposed for the encoding. That results in a 
systematic code
where data and code symbols can be used to reconstruct all data 
symbols.
Symbols are elements in a Galois field and the number 
of finite field operations to encode $k$ code symbols is $\OO(k n)$ and 
decoding all data 
given $k$ code symbols and $n-k$ data symbols takes $\OO(k n)$ finite field 
operations.   In \cite{li2005efficient} an efficient implementation of this 
systematic code is described using the word length of  a processor. In 
\cite{Schindelhauer2013}  the circuit complexity of the Cauchy matrix approach 
was improved to to $(3+o(1)) knw$ operations for encoding $k$ code symbols and 
$9 k n w$ XOR bit operations for decoding from $k$ code symbols for symbol size 
$w$, which is also the asymptotic lower bound for perfect systematic erasure codes \cite{journals/tit/BlaumR99}.

In \cite{Schindelhauer2013} circulant matrices have been used with the word length $w$, where  $w+1$ is a prime number and 2 is a primitive root for $w+1$. 
If $w$ does not have this property, $w$ can be partitioned 
into sub-codes with word size $w_i$, such that $w=\sum_i w_i$ and each $w_i+1$ 
has this property. 

Circulant matrices are equivalent to cyclotomic rings introduced by Silverman~\cite{silverman1999fast} and can be used
for fast multiplication in finite fields. They are used in 
\cite{geiselmann2002new, 
jutla2002circuit} for a small complexity arithmetic circuits for finite 
fields.  Other applications are VLSI implementations  \cite{wu2002finite}, fast 
multiplication \cite{chang2005low}, quantum bit operations 
\cite{amento2012quantum}, and public key cryptography 
\cite{mahalanobis2013matrices}.

The computational complexity of erasure codes can be improved, if one receives more code symbols than data symbols.  An attempt to overcome the limitation 
of the large coding and decoding complexity are Tornado codes {\cite{Byers:1998:DFA:285237.285258, luby2001efficient}, 
where data symbols are encoded by XORs described by a cascaded graph sequence 
combining bipartite sub-graphs.  
For Tornado codes an overhead of $1+\epsilon$ code symbols were produced, which 
could be coded and decoded in time 
$\OO(n \log (1/\epsilon))$.  Starting from this observations, Luby presented LT 
codes in his seminal paper \cite{Luby02}. They use a random set of data symbols 
combined by XOR into the code symbols.
The underlying Robust Soliton distribution has two parameters $c$ and $\delta$, which has been optimized for small 
$n$ in \cite{HTV06}. Luby has shown that, with high probability, LT codes have 
an overhead of only $\OO(\sqrt{n} \log n)$ code symbols to allow 
decoding. Furthermore the coding and decoding time is at most $\OO(n \log n)$ 
parallel XOR operations (multiplied by $w$ for the number of bits in the 
underlying symbols).

Based on LT codes, raptor codes have been introduced  
\cite{shokrollahi2006raptor}. For any given $\epsilon>0$ it is possible to 
recover from $n(1+ \epsilon)$ code symbols to the original $n$ data symbol with 
a complexity of $\OO(\log(1/\epsilon))$ operations for encoding a code symbol 
and $\OO(k \log(1/\epsilon))$ operations to recover $k$ data symbols.  Except 
Cyclone codes, introduced here,  LT codes are the only way to implement raptor 
codes. For this, LT codes have been analyzed in more detail  
\cite{6142082} to improve the behavior of raptor codes. 
\section{Pair Codes}

In \cite{ortolf2009paircoding}  an erasure code called \emph{Pair codes}  is introduced.
The name reflects that each code symbol is a linear combination
of two distinct data symbols, \ie
\[
	y_\ell \coloneqq f_{\ell,1} x_{i_{\ell,1}} + f_{\ell,2} x_{i_{\ell,2}} \eos
\]
The encoder chooses the indexes $i_{\ell,1}, i_{\ell,2}\in [1,n]$ and the 
non-zero coefficients $f_{\ell,1}, f_{\ell,2}$ uniformly at random. The 
coefficients as well as data and code symbols are seen as elements of 
$\GF(2^w)$ with $w\ge 2$, a finite field of size $2^w$. To compute a code 
symbol, one addition 
and two multiplications in $\GF(2^w)$ are required, leading to a total encoding 
complexity of $m$ additions and $2m$ multiplications.

The decoder builds the graph $G_y$ with 
the vertex set $\{x_1, x_2, \ldots, x_n\}$. For each code symbol $y_\ell$,
the graph $G_y$ contains an edge $e_\ell = (x_{i_{\ell,1}},  x_{i_{\ell,2}})$.
The decoder 
uses the following operations:
\begin{itemize}
\item The \emph{Single Rule}: Given an edge $e_\ell$ connecting a decoded 
$x_{i_{\ell,1}}$ with an undecoded $x_{i_{\ell,2}}$, we can decode 
$x_{i_{\ell,2}}$ by
$x_{i_{\ell,2}} = f_{\ell,2}^{-1} ( y_\ell + f_{\ell,1} x_{i_{\ell,1}} )$.
\item \emph{Parallel Edge Resolution}: 
Given two parallel edges $e_k$ and $e_\ell$, \ie
$i_{k,1} = i_{\ell,1}$ and $i_{k,2} = i_{\ell,2}$,
decoding $x_{i_{\ell,2}}$ is possible by solving the system
\[
	y_k \ \ = \ \ 
	f_{k,1} x_{i_{\ell,1}} + f_{k,2} x_{i_{\ell,2}}\qquad\text{and}\qquad
	y_\ell \ \ = \ \ 
	f_{\ell,1} x_{i_{\ell,1}} + f_{\ell,2} x_{i_{\ell,2}}
\]
if the coefficients of $y_k$ and $y_\ell$ are linearly independent.
This yields
\[
	x_{i_{\ell,2}} = (f_{k,1} f_{\ell,2} + f_{k,2} f_{\ell,1})^{-1}
		(f_{k,1} y_\ell + f_{\ell,1} y_k) \eos
\]
\item \emph{Edge Contraction}: Given two adjacent edges $e_k$ and $e_\ell$, \ie 
$i_{k,2} = i_{\ell,1}$, the corresponding code symbols
\[
	y_k \ \ = \ \ 
	f_{k,1} x_{i_{k,1}} + f_{k,2} x_{i_{\ell,1}}\qquad\text{and}\qquad
	y_\ell \ \ = \ \ 
	f_{\ell,1} x_{i_{\ell,1}} + f_{\ell,2} x_{i_{\ell,2}}
\]
can be transformed into a new edge $e'=(x_{i_{k,1}}, x_{i_{\ell,2}})$ and 
a code symbol $y'$ with 
\[
	y' = f_{\ell,1} y_k + f_{k,2} y_\ell =
		f_{k,1}f_{\ell,1} x_{i_{k,1}} + f_{k,2}f_{\ell,2} x_{i_{\ell,2}} \eos
\]
\item The \emph{Double Rule}: Given a cycle $(e_1, e_2, \ldots e_t)$ in $G_y$,
applying the edge contraction rule $t-2$ times to the path $(e_2, e_3, \ldots 
e_t)$ yields an edge $e'$ parallel to $e_1$. Then,   
the parallel edge resolution allows to decode the corresponding data symbols.
\end{itemize} 
One can easily verify that for the single rule one needs to perform one inversion, two multiplications, 
and one addition per data symbol. For the parallel edge resolution one needs to execute one inversion, five 
multiplications, and two additions. The edge contraction only needs four 
multiplications and one addition. 
Considering circuit complexity additions and multiplications can be implemented by linear size and 
constant depth circuits if unbounded MOD2 
gates are available. If we cannot use a look up table for the
inversion, then current implementations, like in~\cite{silverman1999fast},
use linear depth circuits. 
Thus to improve the decoding complexity one has either to minimize the number of inversions or one has to look 
for a coding system, that uses specific values for $f_{i,j}$ which allow 
efficient inversion.

To improve the complexity of the double rule, one has to modify the graph $G_y$ whenever one perform the 
double rule in such a way the the resulting graph reduces its diameter. Within a perfect scenario $G_y$ is a star graph.

The decoder applies the single rule whenever possible. Using this rule,
any connected component of $G_y$ can be decoded, as soon as its first
data symbol is known. As unencoded code symbols are never sent,
the double rule is applied to initiate decoding. It can be applied 
to any connected component of $G_y$ which contains a cycle. 
We note that if the double rule is applied to some cycle $C$, the parallel
edge resolution may fail since the coefficients are linearly dependent.
This is the case if only if $f_{k,1} f_{\ell,2}= f_{k,2}f_{\ell,1}$
which happens with probability $\frac{1}{2^w-1}$ as $f_{k,1}$ is uniformly 
distributed. In this case we can discard one edge of $C$ and the corresponding code symbol, 
as it is a linear combination of the code symbols corresponding to the other
edges of $C$.

To discuss the number of cycles of a connected component we use the notion of 
\emph{excess} of a component, which is the difference between its numbers of 
edges and its number of nodes. The above yields the following result:



\begin{corollary}
The expected number of applications of the double rule per connected component
until it succeeds is $1+\frac{1}{2^w-2}$.
The probability of decoding a connected component $U$ of $G_y$ is 
$1-\left( \frac{1}{2^w-1} \right)^{\mathrm{excess}(U)+1}$.
\end{corollary}


\begin{figure}[ht]\centering
\centering{\small%
\begin{tikzpicture}[x=9mm, y=11mm]
\tikzstyle{process}=[circle, inner sep=0, draw, minimum size=2ex]
\tikzstyle{eblue}=[blue, line width=0.8pt]
\tikzstyle{ered}=[red, line width=0.8pt]
\newcommand{\nodes}[3]{
	\foreach \i in {1,2,...,8} {
		\node[process] at ({\i*360/8-90-360/8-360/16}:1) (c\i) {};
	}
	\path
		(c2) ++(0.75,0) node[process] (e1) {}
		(c3) ++(-15:0.75) node[process] (e2) {}
		(c4) ++(-15:0.75) node[process] (e3) {}
		(c4) ++(15:0.75) node[process] (e4) {}
	;
	\path
		(c1) ++(-90:3ex) node {$x_1$}
		(c2) ++(-90:3ex) node {$x_2$}
		(c8) ++({#1}:3ex) node {$x_5$}
		(e3) ++({#2}:3ex) node {$x_3$}
		(e2) ++({#3}:3ex) node {$x_4$}
	;
};
\newcommand{\edges}{
	\path (c1) edge (c8)
		(c3) edge (e1)
		(c3) edge (e2)
		(c4) edge (e3)
		(c4) edge (e4)
	;
}

\begin{scope}[shift={(0,0)}]
	\nodes{20}{-55}{-55}
	\edges
	\path 
		(c1) edge[eblue] (c2)
		(c2) edge[eblue] (c3)
		(c3) edge (c4)
		(c4) edge (c5)
		(c5) edge (c6)
		(c6) edge (c7)
		(c7) edge (c8)
		;
\end{scope}
\begin{scope}[shift={(4,0)}]
	\nodes{20}{-55}{-55}
	\edges
	\path 
		(c1) edge (c2)
		(c1) edge[eblue] (c3)
		(c3) edge[eblue] (c4)
		(c4) edge (c5)
		(c5) edge (c6)
		(c6) edge (c7)
		(c7) edge (c8)
		;
\end{scope}
\begin{scope}[shift={(8,0)}]
	\nodes{20}{-55}{-55}
	\edges
	\path 
		(c1) edge (c2)
		(c1) edge (c3)
		(c1) edge[eblue] (c4)
		(c4) edge[eblue] (c5)
		(c5) edge (c6)
		(c6) edge (c7)
		(c7) edge (c8)
		;
\end{scope}
\begin{scope}[shift={(12,0)}]
	\nodes{110}{-55}{-55}
	\edges
	\path 
		(c1) edge (c2)
		(c1) edge (c3)
		(c1) edge (c4)
		(c1) edge (c5)
		(c1) edge (c6)
		(c1) edge (c7)
		(c1) edge (c8)
		(c1) edge[eblue, bend left=30] (c8)
		(e2) edge[ered] (e3)
		;
\end{scope}
\end{tikzpicture}}
\caption{Influence of graph modifications on the cycle length.\label{cycle01}}
\end{figure}
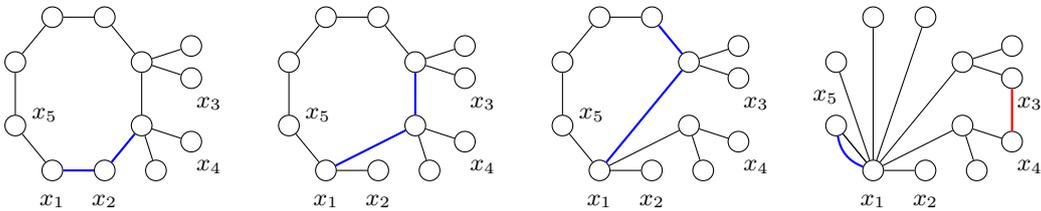

As illustrated in Figure~\ref{cycle01}, performing the double rule and removing 
the correct redundant edge replaced a cycle with a star graph. While 
Pair codes have nearly linear coding and decoding complexity, they suffer from 
the coupon collector problem, as we see in Section~\ref{sec:simulations}.
%

\section{Luby Transform Codes}
\label{sec:lt}


With Luby Transform codes \cite{Luby02}, each code symbol is the bitwise XOR 
of $k_\ell$ distinct data symbols,
where $k_\ell\in [1,n]$ is chosen from a \emph{special} random distribution and
the $k_\ell$ distinct indexes $i_{\ell,j}$ are chosen uniformly at random.
We call $y_\ell$ a clause of size $k_\ell$ and write
\[
	y_\ell = \sum_{j=1}^{k_\ell} x_{i_{\ell,j}}\ .
\]
A basic principle of LT codes is that the size of clauses can be reduced as
more and more data symbols are decoded. Namely, if $x_{i_{\ell,d}}$ for some
$d\in [1,k_\ell]$ has been decoded, then $y_\ell$ can be replaced with the 
following clause of size $k_\ell-1$:
\[
	y' = y_\ell + x_{i_{\ell,d}} = \sum_{j\in \{1,\ldots, k_{\ell}\} 
	\setminus\{d\}}  x_{i_{\ell,j}} \eos
\]
The decoder greedily reduces the size of clauses until the size of a clause is one. Then 
a data symbol has been successfully decoded. However, 
decoding 
cannot start unless the encoder sends clauses of size one. Unlike Pair 
codes, the decoder does not exploit linearly independent clauses of size 2 or 
larger.

The encoding and decoding complexity as well as the overhead depends heavily on
the distribution of the clause sizes. Luby discusses 
two distributions. 

The 
\emph{Ideal Soliton distribution} 
is given by the probability mass function 
$\rho(k)$:
\[
	\rho(k) \coloneqq \begin{cases}
		\displaystyle\frac{1}{n} &
			\text{for $k=1$}\\[2ex]
		\displaystyle\frac{1}{k(k-1)} &
			\text{for $k\in \{2,3,\ldots, n\}$} \eos
	\end{cases}
\]
In the average 
this distribution produces one clause of size one per $n$ code 
symbols and every second code symbol is a clause of size two. 
The 
expected clause size is $H_n$, where $H_n\approx \ln(n)$ is the $n$-th harmonic 
number.

Because of the small number of unary clauses and the small probability that all data symbol are addressed,
the Ideal Soliton distribution works poorly 
\cite{Luby02}. 
Thus, Luby introduces the \emph{Robust Soliton distribution} $\mu$. It is a 
combination of $\rho$ and another distribution $\tau$.
Let $R=c\cdot \ln(n / \delta)\sqrt{n}$, where $c>0$ and $\delta\in(0,1)$ are
tunable parameters, and define
\[
	\tau(k) \coloneqq \begin{cases}
		\displaystyle\frac{R}{kn} &
			\text{for }k=1,2,\ldots, \floor{\frac{n}{R}-1}\\[2ex]
		\displaystyle\frac{R\ln(R/\delta)}{n} &
			\text{for }k=\floor{\frac{n}{R}}\\[1ex]
		0 & \text{otherwise} \eos
	\end{cases}
\]
The Robust Soliton distribution is 
defined as 
\[
	\mu(k) \coloneqq \frac{\rho(k)+\tau(k)}{\beta}
	\quad\text{ with }\quad
	\beta=\sum_{k=1}^{n} \rho(k)+\tau(k)\eoss
\]
where the factor of $1/\beta$ normalizes the probability mass function.

The addition of $\tau$ boosts the probability of clauses with sizes
less than $n/R$. In particular, the expected number of clauses of size
$1$ is increased from $1$ to about $R$ per $n$ clauses.
The expected clause size of the Robust Soliton distribution is 
$\OO(\ln(n/\delta))$ \cite{Luby02} and thus is asymptotically equal to the 
expected clause size of the Ideal Soliton distribution if $\delta$ is constant.
This observation directly implies the following result for constant symbol size 
$w$:

\begin{corollary}
The expected encoding and decoding complexity of an LT code is $\OO(m \ln(n))$ 
XOR word operations using the Ideal Soliton distribution
and $\OO(m \ln(n/\delta))$ XOR word operations using the Robust Soliton 
distribution.

Under the Robust Soliton distribution $m=n+\OO(\sqrt{n} \log n)$ code symbols suffice to decode all data symbols with high probability, \ie $1-n^{-\OO(1)}$. 
\end{corollary}





\newcommand{\px}{\tilde{x}}
\newcommand{\py}{\tilde{y}}
\newcommand{\pad}{\mathrm{pad}}
\newcommand{\unpad}{\mathrm{unpad}}
\newcommand{\zero}{\ensuremath{\mathbb{O}}}
\newcommand{\one}{\ensuremath{\mathbb{I}}}
\newcommand{\two}{\ensuremath{\mathbb{D}}}
\newcommand{\cOne}{\ensuremath{\mathbb{A}}}

\section{Circulant Matrices and Cyclotomic Rings}

We avoid time or memory consuming multiplication operations in finite fields by 
following the approach of cyclotomic rings \cite{silverman1999fast}. These are 
equivalent to circulant matrices, on which Circulant Cauchy 
codes \cite{Schindelhauer2013}, a systematic perfect erasure resilient 
code, are based. Here, we modify this concept for general prime numbers $p$, 
while  
previously, for word size $w$, it was required that $p=w+1$ is a prime number 
with $2$ as a primitive root.

For such $p$, the cyclotomic polynomial of degree $w$,   $\Phi(z) = z^w +\cdots 
+ z^2 + z + 1 \ ,$ is irreducible. The finite field $\GF(2^w)$ is a sub-ring of 
the ring of polynomials modulo $z^p -1$, since $z^p - 1 =  \Phi(z)(z-1)$. In 
order to get efficient multiplication and division operations for  $\GF(2^w)$  
each input $b= (b_0, \ldots, b_{w-1})$ is extended by a so-called {\em ghost 
bit} $b_{w} = 0$. Then, all operations are done in the ring with the extended 
ghost bit basis and retransformed for being output.  This way, Silverman 
\cite{silverman1999fast} reduces the complexity of school multiplication by a 
factor of 2.
 
In  \cite{Schindelhauer2013} an equivalent approach is followed, yet only multiplications and divisions with monoms $b_i z^i$ and binomials $b_i z^i + b_j z^j$ are used. Such operations have linear bit complexity $\OO(w)$ and all operation except the division by a binomial can be computed in constant number of steps by a processor with word length $\Theta(w)$. 

Now let $p=w+1$ be a prime number, where $\Phi(z)$ is not necessarily irreducible. 
We give now a formal description of our operations and the underlying ring $R_p 
= \{0,1\}^p$. We define for $x=(x_0, \ldots, x_{p-1})$, $y=(y_0, \ldots, 
y_{p-1})$ the addition and the multiplication. For $k \in \{0, \ldots, p-1\}$
we have
\[ 
	(x + y)_k = (x_k+y_k) \bmod{2}\ , \quad \quad 
	(x \cdot y)_k = \sum_{i=0}^{p-1} x_i y_{k-i \bmod p} \bmod 2 \eos
\] 
From now on addition on bits is always modulo 2, \ie the XOR operation, 
and $\overline{b}\coloneqq 1+b$ for $b\in\{0,1\}$.
We name the constants $\zero\coloneqq (0,\ldots, 0)$, $\one\coloneqq 
(1,0,\ldots, 0)$, and 
$\two\coloneqq  (0,1,0,\ldots,0)$, and $\cOne\coloneqq (1,\ldots, 1)$, each of 
length $w+1$. Therefore,
$\two^i =(0^{i},1,0^{w-i})$.
So, the multiplication with monomials $\two^i$ in $R_p$, as well its inverse 
operation, is a cyclic shift operation.
The multiplication with binomials 
$\two^i + \two^j$ needs $w+1$ XOR operations. The multiplication with polynomials with $k$ non-zero terms takes $(w+1) (k-1)$ XOR operations. We now deal with the problem of dividing by binomials.
 

We use the following transformation between external $w$-bit representation and internal {\em $p$-bit Ghost Bit Basis} 
representation $R_p$ for $x_i \in \{0,1\}$ and $p=w+1$:
\begin{eqnarray}
\pad(x_0, x_1, \ldots, x_{w-1}) &\coloneqq & (x_0, \ldots, x_{w-1},0) \\ 
\unpad(x_0, x_1, \ldots, x_{w}) &\coloneqq & (x_0 + x_{w}, x_1 + x_{w}, \ldots, 
x_{w-1} + x_{w})\ .
\end{eqnarray}
Under these transformations we get the following elements in the ghost bit base 
for each $x\in \{0,1\}^w$:
\begin{equation}
G(x) \ \ \coloneqq  \ \ \{(x_0, \ldots, x_{w-1}, 0), 
(\overline{x_0}, \ldots, \overline{x_{w-1}}, 1)\} 
\end{equation}
such that $\pad(u) \subseteq G(u)$ and $\unpad(G(x))  =  x$.  
For two sets $S_1,S_2$ we define 
$$S_1 + S_2 \coloneqq  \{s_1 + s_2 \ \mid \  s_1 \in S_1, s_2 \in S_2\} 
\quad\quad\quad
S_1 \cdot S_2: = \{s_1 \cdot s_2 \ \mid \  s_1 \in S_1, s_2 \in S_2\}$$ 
The equivalency class $G$ is closed in $R_p$ under addition and multiplication ($\exists!$ denotes unique existential quantification).
\begin{lemma} 
\label{lem:Gaddmult01} 
We have for all $ x,y \in \{0,1\}^w$
\begin{eqnarray}
 G(x) + G(y) &=& G(x+ y) \label{add} \\
\exists! z\in \{0,1\}^w: \ \ \ G(x) \cdot G(y) & = & G(z) \label{mal}
\end{eqnarray}
\end{lemma}
\begin{proof} For easier notion let $x0$ denote the $w$-bit \emph{symbol} $x$ 
followed 
by 0 and $x1$ denote $x$ followed by 1.\\ 
Equation~(\ref{add}):
Recall that $\overline{x}$ denotes the bitwise negation of the vector $x$ and 
the bitwise XOR is 
denoted by the addition.
\begin{eqnarray*}
G(x) + G(y) &=& \{x0,\overline{x}1\}+\{y0, \overline{y}1\} \\
                   &=& \{(x + y) 0, \overline{(x + y)} 1\} \\
                   &=& G(x+y)
           \end{eqnarray*}
Equation~(\ref{mal}):
First we consider the monomial $y=2^i=(0^i10^{w-i})$ and get
\begin{eqnarray*}
G(x) \cdot G(2^i) & = &  \{(x_0,\ldots, x_{w-1},0), \overline{(x_0,\ldots, 
x_{w-1},0)}\} \cdot 
		\{ \two^i, \overline{\two^i}\} \\
		& = &  \{(x_0,\ldots, x_{w-1},0), \cOne+(x_0,\ldots, x_{w-1},0)\} \cdot 
		\{ \two^i, \cOne+\two^i\} \\
		& = &  \{(x_0,\ldots, x_{w-1},0)\cdot \two^i, \cOne\cdot 
		\two^i+(x_0,\ldots, x_{w-1},0)\cdot \two^i,\\
		& & \ \ (x_0,\ldots, x_{w-1},0)\cdot \cOne+(x_0,\ldots, x_{w-1},0)\cdot 
		\two^i,\\
		& & \ \ \cOne\cdot \cOne+(x_0,\ldots, x_{w-1},0)\cdot \cOne + 
		\cOne\cdot \two^i+(x_0,\ldots, x_{w-1},0)\cdot \two^i\}\ .
\end{eqnarray*}
Note that for even $w$ we have 
\begin{eqnarray*}
\cOne\cdot \cOne & = & \cOne\\
\cOne\cdot \two^i & = & \cOne\\
(x_0,\ldots, x_{w-1},0)\cdot \cOne & = & \left\{\begin{array}[c]{ll}
\zero & \text{if the number of $j$ with $x_j=1$ is even}\\
\cOne & \text{otherwise.}
\end{array}\right.
\end{eqnarray*}
Thus we get
\begin{eqnarray*}
G(x) \cdot G(2^i) & = &  \{(x_0,\ldots, x_{w-1},0)\cdot \two^i, 
\cOne+(x_0,\ldots, x_{w-1},0)\cdot \two^i\}\\
& = &        \{(x_{w-i+1}, \ldots,  x_{w-1} , 0,           x_0,\ldots, 
x_{w-i}), 
	 \overline{(x_{w-i+1}, \ldots,  x_{w-1}}, 1, \overline{x_0,\ldots, 
	 x_{w-i})}\}
\end{eqnarray*}
Hence, $|\unpad(G(x) \cdot G(2^i))| = 1$.  Now consider $y= (y_0, \ldots, 
y_{w-1})$ and by using (\ref{add}) we have 
\[ G(x) \cdot G(y) =   \sum_{i=0}^{w-1} y_i  \cdot G(x)\cdot G(2^i)\ .\]
From $|\unpad(x + y)| =1$ the claim follows by an induction over the number of 
non-zero components of $y$.
\end{proof}

Multiplication with monomial and binomials can be efficiently inverted. Note that this observation cannot be extended to other polynomials unless $2$ is a primitive root for $p=w+1$. 
\begin{lemma}
\label{lem:Gaddmult02} 
We have for all $y \in \{0,1\}^w$,  $i,j \in \{0, \ldots, w\}$, $i\neq j$
\begin{eqnarray}
\exists! x\in \{0,1\}^w: \ \quad\quad \quad\quad\quad  G(x)\cdot \two^i &=& G(y) \label{monomial}\\
\exists! x\in \{0,1\}^w: \  \quad\quad G(x)\cdot (\two^i+ \two^j) &=& G(y) \label{binomial}
\end{eqnarray}
In both cases $x$ can be computed with $w$ bitwise XOR operations.
\end{lemma}
\begin{proof} 
Equation~(\ref{monomial}):
Note that the multiplication with a monomial $\two^i$, \ie a cyclic shift by 
$i$ bits, maps an element of $R_{w+1}$ and its complement to another element 
and its complement. A cyclic shift of $w+1-i$ bits reverses 
this operation. Since $|\unpad(G(y) \cdot \two^{w+1-i})| =1$ and 
$|\unpad(G(x) \cdot \two^i)|=1$ there is no other solution than 
\[ x = \unpad(\pad(y) \cdot \two^{w+1-i}) \ .\]
Equation~(\ref{binomial}):
For  $x', y' \in R_{w+1}$ such that $x=\unpad(x')$, $y=\unpad(y')$ the equation 
$x' (\two^i + \two^j) = y'$ is equivalent to the equations 
\[ 
x'_{(k+i) \bmod (w+1)} + x'_{(k+j) \bmod (w+1)}   =   y'_{k} \ ,\quad 
{\hbox{for $k=0, \ldots, w$}}\]
or equivalently
\[  x'_{(k+ j-i) \bmod (w+1)} =  x'_{k} + y'_{(k-i) \bmod (w+1)}  \ ,\quad 
{\hbox{for $k=0, \ldots, w$}} \ .\]
We choose an element from the class $G(\unpad(x'))$ by setting $x'_{w} = 0$ and 
compute 
\[ x'_{(w+j-i) \bmod (w+1)} = y'_{(w-i) \bmod (w+1)} \ .\]
Then, \[ x'_{(w+2(j-i)) \bmod (w+1)} =  x'_{(w+ (j-i))\bmod (w+1)} + 
y'_{(w+(j-i)-i) \bmod (w+1)}\] and by induction
we can compute for all $k= 1, 2, \ldots, w$
\[x'_{(w+k (j-i)) \bmod (w+1)} =  x'_{(w+ (k-1)(j-i))\bmod (w+1) } + 
y'_{(w+(k-1)(j-i)-i) \bmod (w+1)}\ .\]
If $p=w+1$ is prime then all entries of $x'$ are determined, since we jump 
through all the indices of $x_i$ 
with step size $(j-i)$ modulo $(w+1)$. 
If we would have set $x_{w}=1$ we would have received the complement of $x'$ as 
solution, which 
is the only other solution. Further 
note, that in this process all equations are satisfied. Which proves that  $ x= 
\unpad(x')$ is a solution and the only one.  
Clearly, 
the number of XOR-operations is $p-1=w$.
\end{proof}

With respect to the equivalency class $G$, the multiplication with binomials is 
invertible in 
$R_{w+1}$ and thus we use for the above operation the notation $x = y \cdot (\two^i + \two^j)^{-1}$ for $x,y \in R_{w+1}$.

Consequently, Cyclone codes use the ring multiplication with monomials and 
binomials, since these operations can be implemented as cyclic shifts and 
bitwise XOR operations.


%
%
%
%


\section{Cyclone Codes}
\label{sec:cyclone}

Given a vector of input symbols $x_1, \ldots, x_n \in \{0,1\}^w$ we produce $m 
\geq n$ output symbols $y_1, \ldots, y_m \in \{0,1\}^w$ using a random process. 
We assume that $p=w+1$ is prime.
Recall the Robust Soliton distribution $\mu$ as defined in Section~\ref{sec:lt}.
Each code symbol $y_1, \ldots, y_m$ is constructed as follows.
\begin{algorithm}[H]
 \For{$\ell \leftarrow 1$ \bf{to} $m$}{
  $k_{\ell} \leftarrow $ randomly chosen according the probability mass function $\mu$\\
  $i_{\ell,1}, \ldots, i_{\ell,k_{\ell}} \leftarrow$ randomly, uniformly chosen distinct values from 
  $\{1, \dots, n\}$\\
 $f_{\ell, 1}, \ldots, f_{\ell, k_{\ell}} \leftarrow$ randomly, independently, uniformly chosen from 
 $\{0, \dots, w\}$\\
  $\displaystyle y_{\ell} \leftarrow \unpad\left({\textstyle \sum_{j=1}^{k_\ell}} 
	\two^{f_{\ell,j}} \cdot \pad(x_{i_{\ell,j}})\right)$\\
 }
 \caption{Encoding of Cyclone Codes}
\end{algorithm}
%

\begin{lemma}
\label{lem:cyclone01}
A Cyclone code symbol $y_{\ell}$ can be constructed with expected $w \cdot 
\Expect[k_\ell]$ XOR bit operations.
\end{lemma}
\begin{proof}
 This follows from the observation that unpadding takes $w$ XOR operations and 
 the sum over a clause of $k$ elements takes $(k-1)w$ XOR operations. The 
 cyclic shift and the padding operations do not need any operation at all.
\end{proof}

\subsection*{Decoding}

When reading a code symbol $y_{\ell}\in \{0,1\}^w$ we compute the padded 
version in the ghost bit representation 
$ g_{\ell} = \pad(y_{\ell})$. Then, all computations are done in the ghost bit 
basis and the extra bit is only removed, when the data symbol is to be issued. 
So, for the \mbox{$\ell$-th} code symbol the following information is stored 
$g_{\ell} 
\in R_p$, $k_{\ell}, i_{\ell,1}, \ldots, i_{\ell,k_{\ell}} \in \{1, \ldots, 
n\}$, $f_{\ell, 1}, \ldots, f_{\ell, k_{\ell}}\in \{0, \ldots, w\}$. We call 
this a clause of size $k_{\ell}$ where the following invariant holds
\begin{equation} \label{invariant}
\unpad(g_\ell) = \unpad\left(  {\textstyle \sum_{j=1}^{k_{\ell}}} \two^{f_{\ell, j}} \cdot \pad(x_{ i_{\ell,j}}) \right) \ .
\end{equation} 
On these clauses we implement the following operations:
%

\begin{itemize}
\item \emph{Read code symbol}:
Given $y_{\ell}$, determine $g_{\ell} = \pad(y_\ell)$ and store all parameters 
of the clause.
\item \emph{Output data symbol}:
Given a clause $g_{\ell}$ of size 1 (a monomial), \ie $k_{\ell}=1$, determine 
the
output $x_{\ell,1}$:
\[
 x_{\ell,1}  = \unpad(\two^{p-f_{\ell,1}} \cdot g_\ell) = 
 \unpad(\two^{p-f_{\ell,1}} \cdot \two^{f_{\ell, 1}} \cdot \pad(x_{ 
 i_{\ell,1}})) \eos
\] 
\item \emph{Monomial reduction}:
Given a monomial $g_u$ ($k_u=1$) and a polynomial $g_v$ ($k_v\geq 1$) such that 
$i_{1,u} = i_{j,v}$, we reduce the size of $g_v$ by one by removing $i_{j,v}$ 
and $f_{j,v}$ and replacing $g_v$ with $g'_v = g_v + \two^{f_{j,v}-f_{1,u}} 
\cdot g_u$. 
\item \emph{Parallel edge resolution}: Given two clauses $g_u,g_v$ 
of size $k_{u} = k_{v} = 2$ with $(i_{1,u}, i_{1,v}) = (i_{2,u},i_{2,v})$ and 
$f_{1,u}-f_{2,v} \not\equiv f_{1,u}-f_{2,v} \pmod{p}$, we 
create two monomials $g_{a}$ and $g_{b}$, with $\unpad(g_a)=x_{i_{1,u}}$, 
$\unpad(g_b)=x_{i_{1,v}}$, replacing $g_u, g_v$ as follows:
\begin{eqnarray}
g_a &=&  \left(\two^{f_{1,u}+f_{2,v}} + \two^{f_{1,v}+f_{2,u}}\right)^{-1} \cdot 
         \left(\two^{f_{2,v}} g_1 + \two^{f_{1,v}} g_2 \right) \\
g_b &=&  \left(\two^{f_{1,u}+f_{2,v}} + \two^{f_{1,v}+f_{2,u}}\right)^{-1} \cdot 
         \left(\two^{f_{2,u}} g_1 + \two^{f_{1,u}} g_2 \right) 
\end{eqnarray}
The monomials $g_a$ and $g_b$ replace $g_u$ and $g_v$. If $f_{1,u}-f_{2,v} \equiv f_{1,v}-f_{2,u} \pmod{p}$, then one of the clauses is redundant and can be removed.
\item \emph{Edge Contraction}:
Given two connected clauses $g_u,g_v$ of size $2$, \ie $i_{1,v}=i_{2,u}$,
generate a new clause $g'_2$ with indices $(i'_{2,u}, i'_{2,v}) = (i_{1,u}, 
i_{2,v})$
where 
$f'_{2,u} = f_{1,u} + f_{2,u}$, $f'_{2,v} = f_{1,v}+f_{2,v}$, 
$g'_v = \two^{f_{2,u}} g_u + \two^{f_{1,v}} g_v$. Replace $g_v$ by $g_v'$.
\end{itemize}

\begin{lemma}
\label{lem:cyclone02}
All these clause operations need at most $4 w$ XOR bit operations and preserve 
the 
invariant~(\ref{invariant}).
\end{lemma}
\begin{proof}
This follows from the linear complexity of the padding, unpadding, addition, 
cyclic shift and binomial inverse operation.
\end{proof} 

The first three operations are equivalent to LT codes. Resolving codes this way is called the 
\emph{single rule}.
\begin{lemma} \label{overhead}
Cyclone codes using only monomial reductions are equivalent to LT codes.
\end{lemma}
\begin{proof}
Since the probability distribution for clause lengths is the same, it is 
sufficient to prove that the cyclic shift operations do not influence the 
resolution, which is a straight-forward observation.
\end{proof}

Monomial reduction and parallel edge resolution can be used to resolve large 
sets of data symbols connected by binomials. For this we consider the subset of 
clauses of size $2$ and the corresponding clause graph $G = (V,E)$ with $V=\{1, 
\ldots, n\}$ and edge set of clauses of sizes $1$ and $2$. Loops describe 
clauses of size $1$, where parallel edges are allowed. 
Recall that the excess of a connected component is defined as the difference between the number of edges and nodes in this component.

\begin{lemma}
\label{lem:cyclone03}
If in the clause graph $G$ there exists a connected component with excess 
$c\geq 0$ and $q$ nodes, then all data symbols of this components can be 
resolved with $\OO(qw)$ XOR bit operations with probability $1-1/(w+1)^{c+1}$. 
\end{lemma}
\begin{proof}
If the excess is non-negative, then there exists a cycle $(e_1, \ldots, e_t)$ 
in this component. Using the edge contraction, we can shorten this cycle 
by replacing $e_{1}=\{u_1,u_2\}, e_{2}=\{u_2, u_3\}$ within the cycle 
by edge edge $e'_1= \{u_1, u_3\}$. After $t-2$ edge contractions we have reduced the 
cycle to a cycle of length $2$ consisting of two parallel edges. These two 
parallel edges can be resolved using the parallel edge resolution if 
$f_{1,1}-f_{2,2} \not\equiv f_{1,2}-f_{2,1} \pmod{p}$. Note that all operations 
uphold the uniform probability distribution for the cyclic shift factors. 
Hence, this equality does not holds with probability $\frac{1}{w+1}$. 
In this case 
we remove one edge of the considered cycle and thus we decremented excess of the 
component. If the new excess is 
still non-negative we can restart the process with a new cycle.

Otherwise, if $f_{1,1}-f_{2,2} \not\equiv f_{1,2}-f_{2,1} \pmod{p}$
we resolve $u_1$ and $u_2$ which allows us to resolve 
all data symbols of the connected components repeatedly using the monomial 
reduction.

Depending on the way we choose the edges within the edge contraction we can bound
the number of edge contractions by the number of nodes in 
the components plus the number of linear dependent edges, if we transform 
each cycle to a star.
\end{proof}

The resolution of connected components in the clause graph is called the 
\emph{double rule}. The decoding algorithm 
applies the single and double rule as long as they are successful. 
\begin{algorithm}[H]
  \For{$\ell \leftarrow 1$ \bf{to} $m$}{
    Read code symbol $y_{\ell}$\\
    \Repeat{code set not changed in this round}{
  	\Repeat {
 		Single Rule does not change clauses
	}{Apply Single Rule}
	\Repeat {
 		Double Rule does not change clauses
	}{Apply Double Rule}
	}
	}
 \caption{Cyclone Decoding}
\end{algorithm}

It is not fully understood why the double rule improves efficiency. An explanation is that the clause graph is an Erd\"o{s}-Renyi-Graph \cite{erdHos1959random}. For $n/2$ edges a giant component of sizes $\Theta(n^{2/3})$ appears and grows dramatically for increasing number of edges, where for every random edge four elements are added to the component \cite{bollobas1984evolution}. There it is also shown that for $n/2 + \omega(n^{2/3} \log^{1/2} n)$  edges the excess of the giant component is  non-negative, which allows the decoding of this component. Understanding this process in combination of the effect on clauses larger than $2$ might help to build better probability distributions for the Cyclone codes as the Robust Soliton distributions already allow. For the efficiency and coding overhead the following observation can be shown.

\begin{theorem}
Cyclone codes need at most the asymptotic time complexity and coding overhead 
as LT codes.
\end{theorem}
\begin{proof}
The coding overhead follows from Lemma~\ref{overhead}, where we have seen that LT codes are included in the first three rules.

For the time complexity, we have seen that applying the pair rule adds only a linear number of operations for each decoded data symbol. Most time is still consumed by the single rules, where longer clauses are stripped from decoded data symbols. An operation which also takes place in LT codes, yet with some additional code symbols.
\end{proof}


\section{Fermat, Goldbach and the Word Length}

All operations except the binomial division can be performed word parallel, such that $w$ operations 
can be done in constant number of processor steps. It turns out that the binomial division is only 
needed once for each connectivity component, so its influence is marginal.

It is of particular interesting to use powers of $2$ as the word length $w$.
However, since $w+1=p$ is required to be a prime number, this would restrict us 
to Fermat primes of form $2^{2^i}+1$. They are the only numbers of form 
$2^j+1$ which can be prime.  Only $5$ Fermat primes are known, corresponding to 
word lengths $2$, $4$, $16$, $256$, and $65\,536$.

In \cite{Schindelhauer2013}, a small trick is introduced which generalizes the 
word lengths to all even positive numbers. For this we split a data symbol of 
length $w$ into two separate encodings of lengths $w_1$ and $w_2$, where 
$w_1+1$ and $w_2+1$ are prime numbers. This is possible if $w+2$ can be 
represented as sum of two primes $p_1, p_2$, which is the still open Goldbach 
conjecture, one of Hilbert's eight problems \cite{hilbert1900mathematische}. It 
is known that this is the case for $w \leq 4 \cdot 10^{14}$ 
\cite{richstein2001verifying}. Then, we send the code symbols 
in both encodings without overhead. However, we can only decode if both codes can be decoded, which might lead to an increased overhead, unless $w_1=w_2$. Then we can use the same parameters, twice. Because of the probabilities depending on the word 
length, the best choice is to choose $w_1$ and $w_2$ of nearly equal size, 
which leads to the following word lengths, \eg
\[
\begin{array}{rclcrclcrcl}
8      &=& 4+4      & \quad & 32  &=& 16+16       & \quad &  64  &=& 28+36 \\
128  &=& 58+70  & & 512 &=& 256+256 & & 1024 &=& 502 + 522 \eos
\end{array}
\]

\section{Simulations}
\label{sec:simulations}


We have run extended simulations to estimate the overhead of the Cyclone code. 
For this we generate a series of $s$ code symbols and count the number of 
decoded 
data symbols. We compare Cyclone codes using the Ideal Soliton distribution and 
the 
Robust Soliton distribution for the clause size with LT codes on the same 
clause size
distributions. Furthermore, we show the behavior of uniformly chosen random data symbols 
(suffering under the coupon collector problem) and Pair codes, where the pairs are chosen uniformly.

Figure~\ref{run100} shows the number of decoded data symbols (vertical axis),  
for a growing series of $s$ coded symbols for the above mentioned codes. The 
number of overall data symbols is $n=100$. The word size is $w=256$. For the 
Robust Soliton distribution we chose as parameters $\delta=0.5$ and $c=0.01$  
\cite{MacKay05}. The straight lines represent the average over 1000 tests. The 
dotted lines show the 90\%-percentile of the set of decoded data symbols from 
1000 runs with different random numbers.  In Appendix~\ref{app:simulations}, we 
show the corresponding runs for $n=1000$.

\begin{figure}[t]\centering
\includegraphics[width=\textwidth]%
	{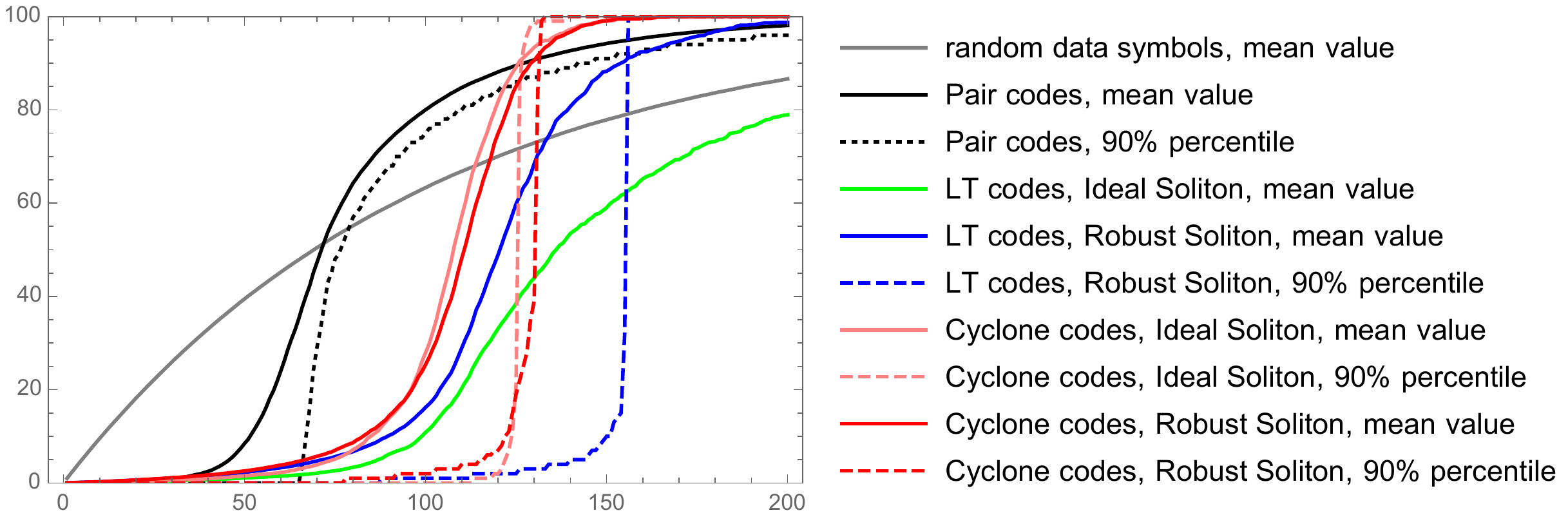}
\caption{The average number and 90\%-percentile of the number of decoded symbols with respect 
to the number of available fountain codes for 100 data symbols and 1000 test
samples.\label{run100}}
\end{figure}

The probability that a data symbol is not covered after $m$ code symbols 
is described by 
$(1-\frac{1}{n})^m$, which can be approximated by $e^{-m/n}$ for large enough $n$. Hence, 
the expected number of available symbols follows the function $n (1-e^{-m/n})$, as the simulations 
clearly show. For Pair codes we see only few decoded symbols before $m=n/2$. At 
$m =n/2$ 
the random clause graph shows the appearance of the giant component of size $\Theta(n^{2/3})$. 
So, the probability of a cycle in such a component tends towards $1$ and the decoding starts. 
In the long run, Pair codes suffer from a reduced coupon-collector problem, 
where the upper limit 
of the function is $n(1-e^{-2m/n})$, since the probability that a data node is not covered in 
the clause graph is $(1-\frac{2}{n})^m$.

It is known that LT codes using the Ideal Soliton distribution do not perform 
well. For 
$n=100$ and $n=1000$ it behaves worse than sending random data symbols. Hence, it is quite 
surprising that for $n=100$ a Cyclone code performs even better than Cyclone codes with respect 
to the Robust Soliton distributions. In the median 118 code symbols (overhead 18\%) are sufficient 
and for only 10\% of the samples more than 136 code symbols were necessary. Cyclone codes with Robust Soliton distribution needs more than 138 code symbols for only 10\% of the samples and 125 in the 
median. The overhead is calculated as the relative number of extra code symbols in order to decode 
the data symbols, \ie $\frac{m-n}{n} = \frac{m}{n}-1$. For LT codes it is known 
that they have 
considerable overhead for such small number of data symbols. It is recommended 
to use at least $n \geq 10\,000$.

Figure \ref{exp-18} takes this into account by increasing the number of data 
symbols $n=2^i$ for $i=1, \ldots, 18$, \ie $n=2,4,8, \ldots, 262\,144$, we have 
performed $1000$ tests with a series of fountain code symbols, where we 
stopped the test each time at the $m$-th code symbol, when all data symbols 
could be computed. Then, we computed the overhead $m/n -1$. For random data 
symbols, Pair codes, LT codes with Ideal and Robust Soliton 
distribution  
($\delta=0.5$, $c=0.01$), and Cyclone codes with Ideal and Robust 
Soliton 
distribution we plotted the average overhead and the 90\%-percentile in a 
log-log-plot. The word size was chosen as $w=256$.

\begin{figure}[t]\centering
\includegraphics[width=\textwidth]%
	{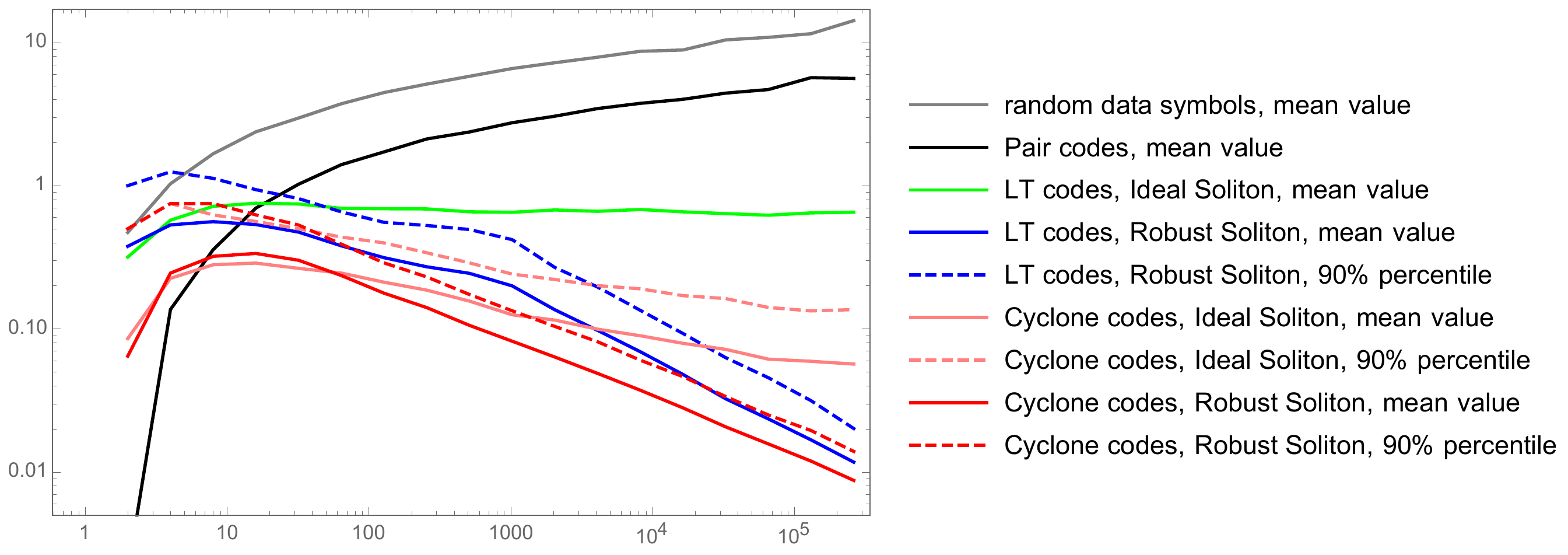}
\caption{A log-log-plot of the average (and $90\%$-percentile) relative overhead 
of extra code symbols to be sent in different fountain codes for increasing number 
of data symbols.\label{exp-18}}
\end{figure}

For small $n< 8$ Pair codes performs best, then the Cyclone code with Ideal Soliton distribution 
takes over until for $n\geq 128$ the Cyclone code with Robust Soliton consistently outperforms all 
other displayed codes. For $n=8192$ Cyclone codes with a Robust Soliton distribution have a median 
overhead of $3.4\%$, an average overhead of $3.9\%$, a standard deviation of 
$1.5\%$, and a 
$90\%$-percentile of $5.9\%$. The corresponding values for LT codes are a median overhead of $4.5\%$, 
an average overhead of $5.8\%$, a standard deviation of $3.5\%$, and a $90\%$-percentile of $10.3\%$.

\section{Conclusions}

Cyclone codes combine LT codes \cite{Luby02} with Pair codes 
\cite{ortolf2009paircoding} and decrease the message overhead without 
increasing the coding and decoding complexity. The coding overhead is $(m-n)/n$ 
where $m$ the number of code symbols to be received to decode $n$ data symbols. 
By the nature of XOR operations, LT codes can not benefit from the giant 
component in the random clause graph, \ie the graph described by code symbols 
combining two data symbols. There, every cycle corresponds to redundant code 
symbols, which can be immediately discarded. Pair codes solely rely on such 
complex components. 

While Pair codes only need a linear number of Galois field operations for 
coding and decoding, they suffers 
from the coupon collector problem, such that the coding overhead is $\Theta(\log n)$. For small $n$ it 
outperforms the other coding schemes, but for such cases systematic perfect erasure codes are still efficient 
enough. An example of such a systematic perfect efficient erasure resilient scheme with small number of XOR 
operations is the Circulant Cauchy code \cite{Schindelhauer2013}, where Boolean circulant matrices, also known 
as cyclotomic rings \cite{silverman1999fast}, have been used because of its 
efficiency.  In such rings the word length is restricted to numbers $w$, where 
$w+1$ is prime with 2 as a primitive root, \ie all numbers $2^0, 2^1, \ldots 
2^{w}$ modulo $w+1$ are distinct. It is unknown how many such prime number 
exist, while they do appear to be quite often in prime numbers. This 
restrictions has been resolved in  \cite{Schindelhauer2013} by partitioning $w= 
\sum_{i} w_i$ such that each $w_i+1$ is prime with 2 as a primitive root.

Here we drop the condition of $w+1$ having 2 has a primitive root, since we observe that the unitary 
cyclotomic ring has some elements that can be inverted, even if it is not isomorphic to a finite field. 
So, we extend the set of word lengths being powers of two from $w\in\{2,4\}$ to 
all numbers where $w+1$ are Fermat primes, namely $2$, $4$, $16$, $256$, and 
$65\,536$. Still the partitioning technique works as well for all even $w$, 
such 
that Cyclone codes exist for all even $w$ using $w_1+w_2=w$ if the Goldbach 
conjecture holds. Each read, write, addition, multiplication and division 
operation consist only of at most one cyclic shift operations and at most $w+1$ 
bitwise XOR operations. It also allows the usage of the word parallelism of 
processors, especially since simulations show that the coding overhead does not 
increase significantly for small $w$. Clearly, Cyclone codes can also be 
implemented using finite fields and general factors. However, we do not 
expect any significant benefit.

\paragraph*{Outlook}
Raptor codes are based and on LT codes without any alternative. The situation 
has changed with the presentation of Cyclone codes. Since they outperform LT 
codes, it is straight-forward to look at a modified Raptor 
code based on Cyclone codes.

Another open area of research are special probability distributions optimized 
for Cyclone codes. At the moment we have only used the Ideal and Robust Soliton 
distribution, both of which are optimized for the single rule of Cyclone 
codes. However, the double rule harnesses the complex components of the random 
graph. Here lies the potential of even less coding overhead. Yet, the behavior 
of complex connected components in the 2-clause random graph, dynamically 
changed by the ripple effect of some Soliton distribution, is poorly 
understood. The investigation of such dynamic random graphs, where complex 
components are removed while edges are added to the residual graph, is crucial 
for improving Cyclone codes.

\paragraph*{Acknowledgements}
We thank Christian Ortolf for his valuable input and many fruitful discussions. 

Parts of this work have been supported by the {\em Sustainability Center 
Freiburg}.


\bibliographystyle{plain}
\bibliography{tr}

\newpage
\appendix

\section{Additional Simulation Results}
\label{app:simulations}

Figures~\ref{run10},  \ref{run1000}, and \ref{run10000}  show the number of 
decoded data symbols (vertical axis),  
for a growing series of $m$ code symbols  for the above mentioned codes. The 
number of overall data symbols is $n=10$ in Figure~\ref{run10},  $n=1000$ in 
Fig.~\ref{run1000}, and $n=10\,000$ in Fig.~\ref{run10000}. The word size is 
$w=256$. For the 
Robust Soliton distribution we chose as parameters $\delta=0.5$ and $c=0.01$  
\cite{MacKay05}. The straight lines represent the average over $1000$ tests. 
The 
dotted lines show the 90\%-quantile of the set of decoded data symbols from 
$1000$ runs with different random numbers.

\begin{figure}[H]\centering
\includegraphics[width=\textwidth]%
	{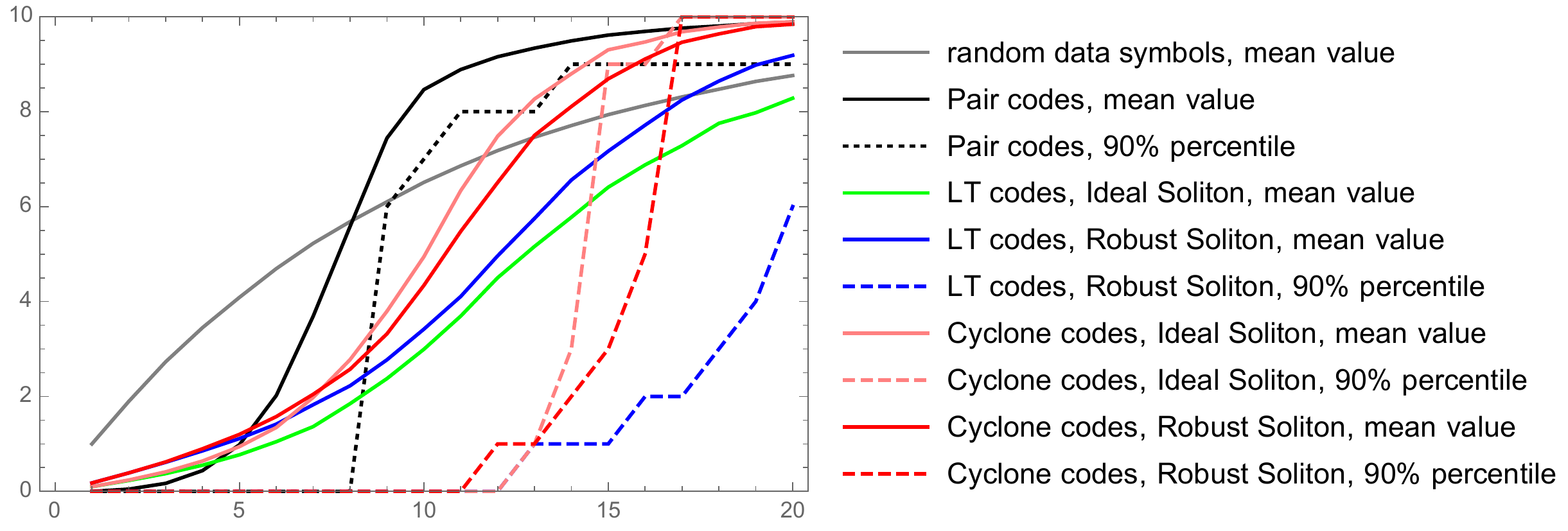}
\caption{The average number and 90\% quantile of the number of decoded symbols 
with respect to the number of available fountain codes for 10 data symbols and 
$1000$ test samples.\label{run10}}
\end{figure}

\begin{figure}[H]\centering
\includegraphics[width=\textwidth]%
	{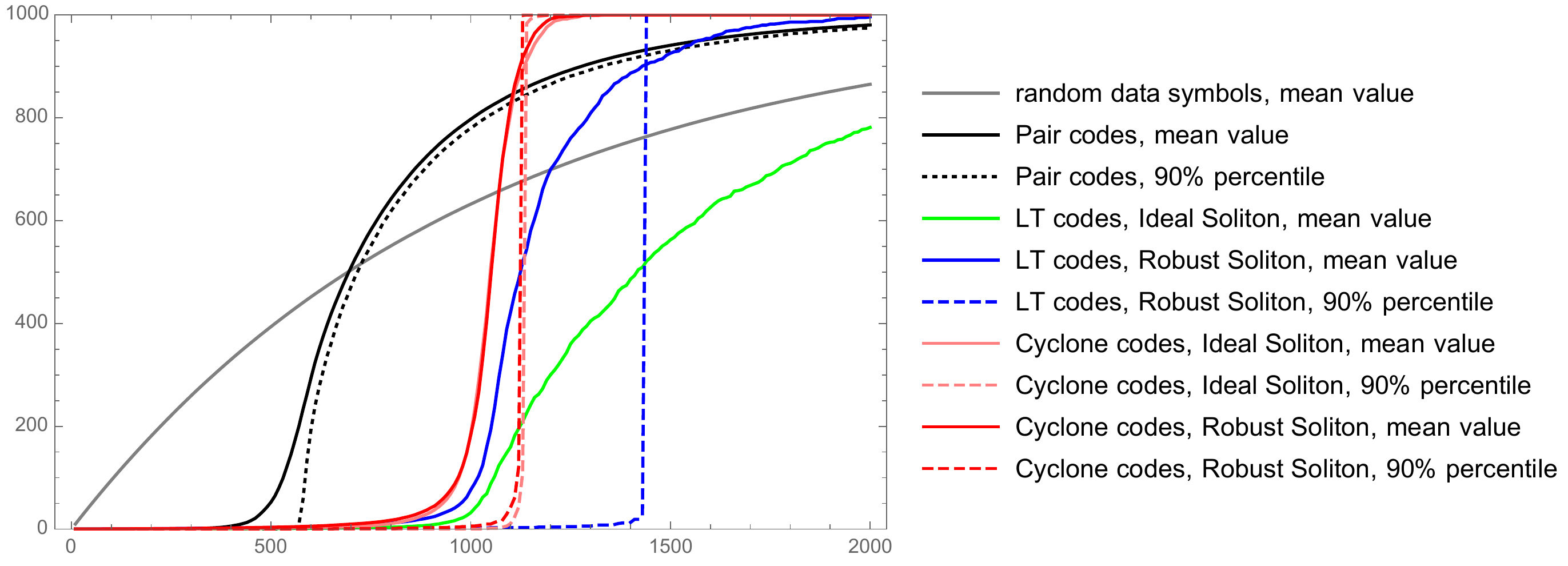}
\caption{The average number and 90\% quantile of the number of decoded symbols 
with respect to the number of available fountain codes for $1000$ data symbols 
and $1000$ test samples.\label{run1000}}
\end{figure}

\begin{figure}[H]\centering
\includegraphics[width=\textwidth]%
	{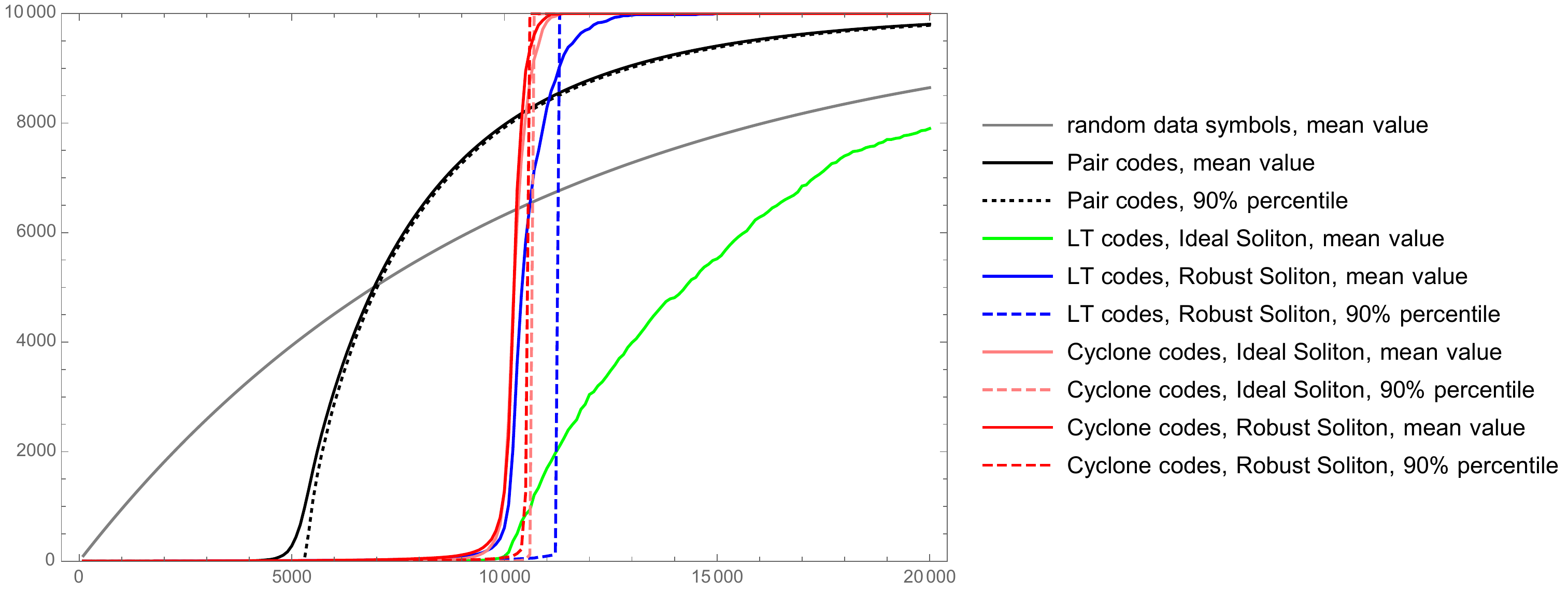}
\caption{The average number and 90\% quantile of the number of decoded symbols 
with respect to the number of available fountain codes for $10\,000$ data 
symbols and $1000$ test samples.\label{run10000}}
\end{figure}

Figure \ref{Rsols} compares different parameters for the Robust Soliton 
distribution, \ie  $\delta=0.5$ and $c=0.01$ and $c=0.03$.   The number of data 
symbols is $n=2^i$ for $i=1, \ldots, 18$, \ie $n=2, 4, 8, \ldots, 262\,144$,
for  
tests. The test run increases the code size $m$ until all $n$ data symbols can 
be computed, then the overhead $m/n-1$ is computed. The word size is $w=256$. 
For random data symbols, Pair codes, LT codes with Robust Soliton distribution 
and Cyclone codes with Ideal Soliton and Robust Soliton distribution the 
average 
overhead is displayed  in a log-log-plot.

\begin{figure}[H]\centering
\includegraphics[width=\textwidth]%
	{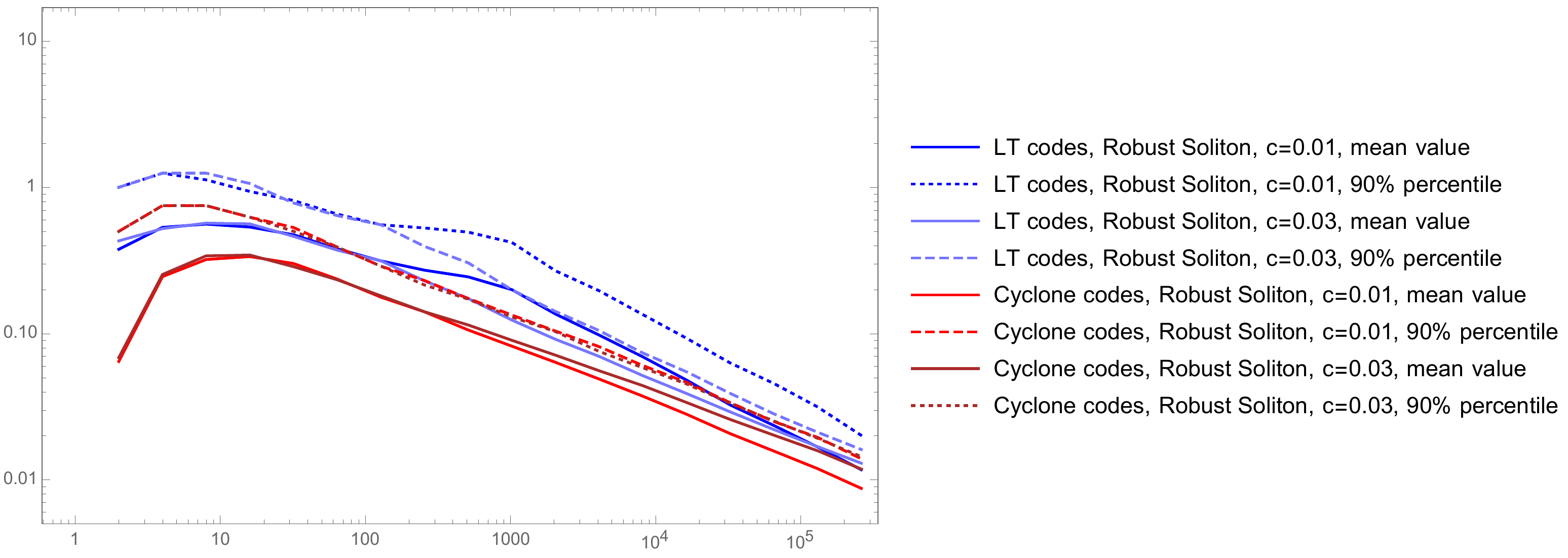}
\caption{The average overhead of Cyclone and LT codes for increasing number of data symbols for $\delta=0.5$ and $c\in\{0.01,0.03\}$.\label{Rsols}}
\end{figure}


%

Figures~\ref{fig:hist10}, \ref{fig:hist100}, \ref{fig:hist1000}, 
and~\ref{fig:hist10000} show histograms of the number of code symbols required 
to successfully decode $n=10$, $n=100$, $n=1000$, and $n=10\,000$ data symbols, 
respectively. Each histogram is the result of $1000$ 
trials, where both LT codes and Cyclone codes use the Robust Soliton 
distribution with $\delta=0.5$ and $c=0.01$. The word size is $w=256$.
Given the same amount of code symbols, Cyclone codes are able to decode
the complete set of data symbols significantly more often than LT codes.

\begin{figure}[H]
\small\flushleft
\includegraphics[page=1]{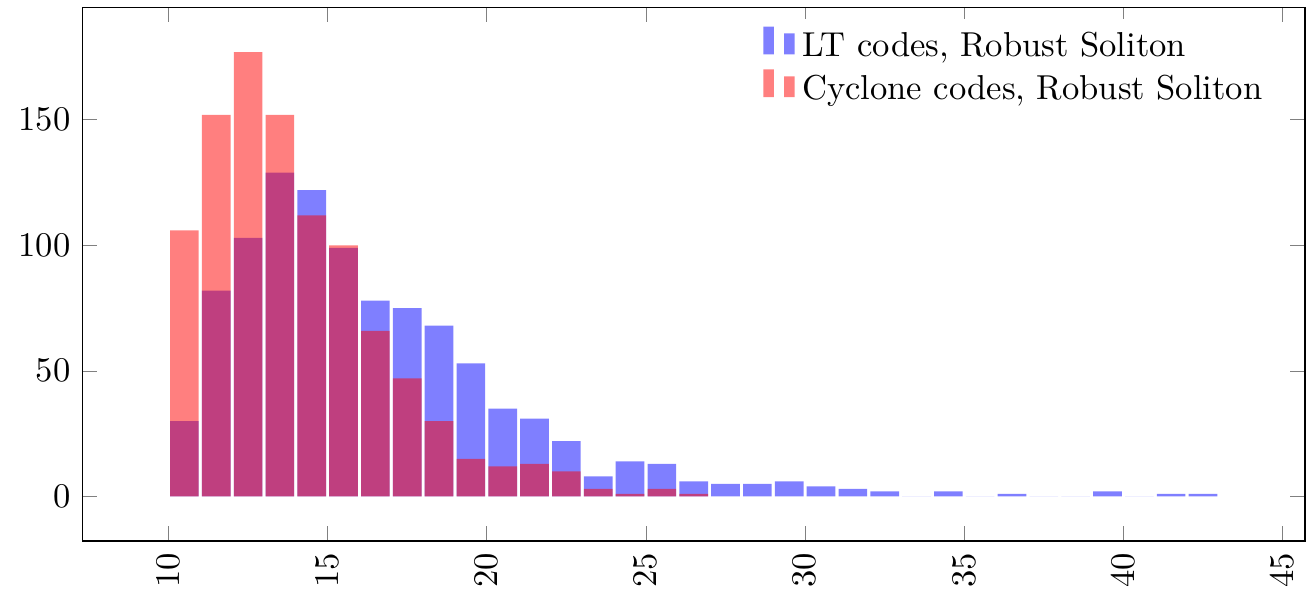}
\caption{Histogram of the number of code symbols required to decode
	all $n=10$ data symbols. The total number of trials is $1000$ with 
	parameters $\delta=0.5$ and $c=0.01$ for the Robust Soliton distribution.}
\label{fig:hist10}
\end{figure}

\begin{figure}[H]
\small\flushleft
\includegraphics[page=2]{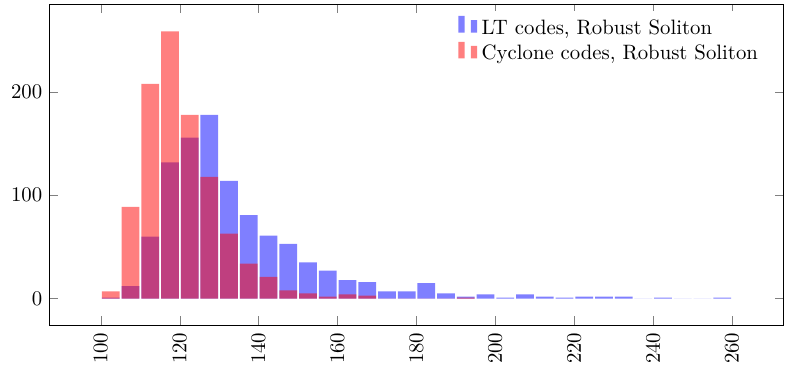}
\caption{Histogram of the number of code symbols required to decode
	all $n=100$ data symbols. The total number of trials is $1000$ with 
	parameters $\delta=0.5$ and $c=0.01$ for the Robust Soliton distribution.}
\label{fig:hist100}
\end{figure}

\begin{figure}[H]
\small\flushleft
\includegraphics[page=3]{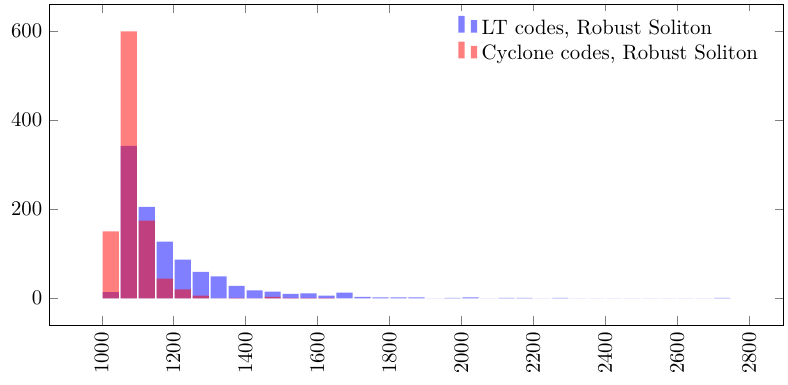}
\caption{Histogram of the number of code symbols required to decode
	all $n=1000$ data symbols. The total number of trials is $1000$ with 
	parameters $\delta=0.5$ and $c=0.01$ for the Robust Soliton distribution.}
\label{fig:hist1000}
\end{figure}

\begin{figure}[H]
\small\flushleft
\includegraphics[page=4]{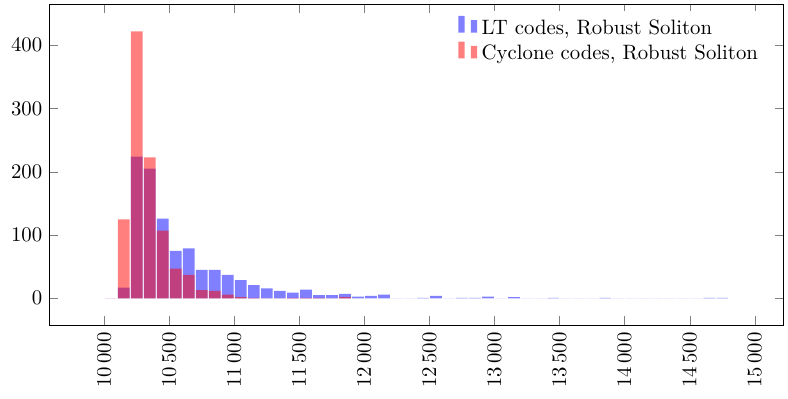}
\caption{Histogram of the number of code symbols required to decode
	all $n=10\,000$ data symbols. The total number of trials is $1000$ with 
	parameters $\delta=0.5$ and $c=0.01$ for the Robust Soliton 
	distribution.}
\label{fig:hist10000}
\end{figure}

\end{document}